\documentclass{LMCS}%

\def\doi{8 (2:15) 2012}
\lmcsheading%
{\doi}
{1--10}
{}
{}
{May\phantom.~17, 2011}
{Jun.~25, 2012}
{}
\usepackage{enumerate,hyperref}
\usepackage{amssymb}%
\usepackage{amsmath}%
\usepackage{amsfonts}%
\usepackage{graphicx}
\providecommand{\U}[1]{\protect\rule{.1in}{.1in}}
\begin{document}

\title[Precompact Apartness Spaces]{Precompact Apartness Spaces}
\author[D.~S.~Bridges]{Douglas S. Bridges}
\address{Department of Mathematics \& Statistics, University
of Canterbury, Private Bag 4800, Christchurch 8140, New Zealand}
\email{d.bridges@math.canterbury.ac.nz}

\subjclass{G.m}

\begin{abstract}%
\noindent
We present a notion of precompactness, and study some of its properties, in
the context of apartness spaces whose apartness structure is not necessarily
induced by any uniform one. The presentation lies entirely with a Bishop-style
constructive framework, and is a contribution to the ongoing development of
the constructive theories of apartness and uniformity.

\end{abstract}%

\maketitle

\section{Introduction}

The apparent difficulty of producing a decent notion of \emph{compactness} for
the constructive\footnote{For us, \emph{constructive mathematics} means mathematics in the style of
Errett Bishop: that is, mathematics developed with intuitionistic logic and
some foundational system such as the constructive set theory in
\cite{Aczel1,Aczel2} or the type theory in \cite{ML}. We shall assume that the
reader is familiar with, or has access to, the fundamentals of constructive
analysis, as found in \cite{Bishop,BB,bvtech} (see also \cite{Beeson,TvD})}
theory of pre-apartness spaces has provoked several responses
\cite{dsb,Diener,Steinke}. In this paper we offer a partial
response---precompactness---without claiming, or even suggesting, that it may
be the \textquotedblleft right\textquotedblright\ one: we prefer to leave our
offer on the table, to be considered further and then judged alongside those
of other mathematicians who have considered the compactness problem.

Let $X$ be an inhabited set with an inequality relation $\neq$, and let
$\bowtie$ be a relation between subsets of $X$ that satisfies certain axioms
to be stated shortly. We define three types of complement for a subset $S$ of
$X$:%
\begin{align*}
\lnot S  & \equiv\left\{  x\in X:x\notin S\right\}  ,\\%
\mathord{\sim}%
S  & \equiv\left\{  x\in X:\forall_{s\in S}\left(  x\neq s\right)  \right\}
,\\
-S  & \equiv\left\{  x\in X:\left\{  x\right\}  \bowtie S\right\}  .
\end{align*}
For $x\in X$, we normally write $x\bowtie S$ rather than $\left\{  x\right\}
\bowtie S$. We call $\bowtie$ a \textbf{pre-apartness} on $X$ if it satisfies
these four axioms:

\begin{enumerate}[\bf(B1)]
\item $X\bowtie\varnothing$

\item $-A\subset%
\mathord{\sim}%
\,A$

\item $\left(  \left(  A_{1}\cup A_{2}\right)  \bowtie(B_{1}\cup
B_{2})\right)  \Leftrightarrow\forall_{i,j\in\left\{  1,2\right\}  }%
A_{i}\bowtie B_{j}$

\item $-A\subset%
\mathord{\sim}%
B\Rightarrow-A\subset-B$.
\end{enumerate}%

\noindent
Taken with the relation $\bowtie$, the set $X$ then becomes a
\textbf{pre-apartness space. }If, in addition, $\bowtie$ satisfies
\footnote{\textbf{(B4)} can actually, and easily, be derived from \textbf{(B1)}--\textbf{(B3)}
and \textbf{(B5)}; see page 69 of \cite{bvbook}.}

\begin{enumerate}[\bf(B1)]
\item[\bf(B5)] $x\in-A\Rightarrow\exists_{S\subset X}(x\in-S\wedge
X=-A\cup S)$,
\end{enumerate}%

\noindent
then it is called an \textbf{apartness} on $X$, and $X$ becomes an
\textbf{apartness space.}

There is a natural topology associated with a pre-apartness space: the
\textbf{apartness topology}, in which the apartness complements form a base of
open sets. It is this topology that we refer to when we consider notions like
\emph{open}, \emph{closed}, and \emph{dense} below.

Perhaps the most important example of an apartness space is given by a uniform
space: that is, an inhabited set $X$ with a uniform structure $\mathcal{U}%
$.\footnote{We omit the detailed definition of \emph{uniform structure}, referring the
reader to Section 3.2 of \cite{bvbook} for more information about apartness
and uniform spaces.} In this case, the apartness is defined, for subsets $S,T
$ of $X$, by%
\[
S\bowtie T\Leftrightarrow\exists_{U\in\mathcal{U}}(S\times T\subset%
\mathord{\sim}%
U).
\]
The apartness topology corresponding to this apartness relation coincides with
the usual uniform topology on $X$. Note that, in contrast to the classical
situation, for us a uniform space is Hausdorff, by definition. Moreover, the
inequality relation defined on $X$ by%
\[
x\neq y\Leftrightarrow\exists_{U\in\mathcal{U}}\left(  \left(  x,y\right)
\notin U\right)
\]
is \textbf{tight}:%
\[
\lnot\left(  x\neq y\right)  \Rightarrow x=y\text{.}
\]
Since the uniform/apartness topology is Hausdorff and the inequality is tight,
convergent nets in a uniform space have unique limits (\cite{bvbook},
Corollary 2.4.8).

A particular case occurs when $\left(  X,\rho\right)  $ is a metric space,
when the uniform structure $\mathcal{U}$ has a base comprising sets of the
form%
\[
\left\{  \left(  x,x^{\prime}\right)  \in X\times X:\rho\left(  x,x^{\prime
}\right)  <\varepsilon\right\}  ,
\]
where $\varepsilon>0$; then%
\[
S\bowtie T\Leftrightarrow\exists_{\varepsilon>0}\forall_{s\in S}\forall_{t\in
T}\left(  \rho\left(  s,t\right)  \geq\varepsilon\right)  .
\]

A uniform space $\left(  X,\mathcal{U}\right)  $ is said to be \textbf{totally
bounded} if for each $U\in\mathcal{U}$ there exist finitely many points
$x_{1},\ldots,x_{n}$ of $X$ such that%
\[
X=%
{\displaystyle\bigcup\limits_{i=1}^{n}}
\left\{  y\in X:\left(  x,y\right)  \in U\right\}  .
\]
In the case of a metric space $X$, this reduces to the familiar elementary
property of total boundedness. The closure of a totally bounded subset of a
uniform space is totally bounded (\cite{bvbook}, Proposition 3.3.5). Moreover,
as is straightforward to show, a dense subspace of a totally bounded uniform
space is totally bounded.

A mapping $f:X\rightarrow Y$ between pre-apartness spaces is said to be
\textbf{strongly continuous}\ if\footnote{Strongly continuous mappings are the morphisms in the category of
pre-apartness spaces.}%
\[
\forall_{A,B\subset X}\left(  f(A)\bowtie f(B)\Rightarrow A\bowtie B\right)
.
\]
If $(X,\mathcal{U)}$ and $\left(  Y,\mathcal{V}\right)  $ are uniform spaces,
$f:X\rightarrow Y$ is strongly continuous, and $f(X)$ is totally bounded, then
$f$ is \textbf{uniformly continuous} in the sense that%
\[
\forall_{V\in\mathcal{V}}\exists_{U\in\mathcal{U}}\forall_{x,x^{\prime}\in
X}\left(  \left(  x,x^{\prime}\right)  \in U\Rightarrow(f(x),f(x^{\prime})\in
V\right)
\]
(\cite{bvbook} Theorem 3.3.9). Uniformly continuous mappings between uniform
spaces are strongly continuous and preserve total boundedness (\cite{bvbook},
Propositions 3.3.2 and 3.3.7).

In the remainder of this paper \emph{we shall assume that our pre-apartness
spaces are symmetric}: if $A\bowtie B$, then $B\bowtie A$. Uniform apartness
spaces are symmetric.

\section{Precompact spaces}

We approach precompactness by reminding the reader of two theorems, the first
of which is a consequence of Theorem (1.4) in Chapter 5 of \cite{BR}.

\begin{thm}
\label{one}If $X$\ is a \textbf{compact}---that is, inhabited, complete, and
totally bounded---metric space, then there exists a uniformly continuous,
open\footnote{Recall that an open mapping is one that preserves the openness of sets. In
fact, the theorem in \cite{BR} proves more: there exists a \emph{uniform
quotient map}, a stronger notion than that of a uniformly continuous, open
map, from $2^{\mathbf{N}}$ onto $X$.} mapping of $2^{\mathbf{N}}$ onto $X.$
\end{thm}%

\noindent
The metric-space case of the second was given by Ishihara and Schuster
\cite{Isch}; the proof of the following general case can be found in
\cite{bvbook} (Theorem 3.3.18).

\begin{thm}
\label{two}Let $X$ be a totally bounded uniform space with a countable base of
entourages. Then every strongly continuous mapping from $X$ into a uniform
space is uniformly continuous.
\end{thm}%

\noindent
In particular, every strongly continuous map from Baire space $2^{\mathbf{N}}
$ into a uniform space is uniformly continuous.

With these two theorems in mind, we define an apartness space $X$ to
be\textbf{\ precompact} if there exists a strongly continuous mapping from a
dense subset of $2^{\mathbf{N}}$ onto a dense subset of $X$.\footnote{In classical topology, the name \emph{precompact} is applied sometimes to sets
that are dense a compact space, and sometimes to totally bounded sets in a
uniform space. It is the latter application that we have in mind in making our
current definition.
\par
It is tempting to define precompactness in terms of images of arbitrary
totally bounded uniform spaces under strongly continuous mappings. However, in
the classical theory of proximity spaces, of which apartness spaces are the
constructive counterpart, every proximity space satisfying the Efremovi\v{c}
property (for which see later) would be precompact in that sense, since it
would have a totally bounded uniform structure compatible with its
proximity/apartness structure; see \cite{Naimpally} (page 72, (12.3)).} An
apartness space with a dense, precompact subset is itself precompact. The
image of a precompact apartness space under a strongly continuous mapping is precompact.

\begin{prop}
\label{AA1}A precompact uniform space is totally bounded.
\end{prop}

\begin{proof}
Let $X$ be a precompact uniform space. Then there exist a dense subset of
$2^{\mathbf{N}}$ and a strongly continuous mapping $h$ of $D$ onto a dense
subset of $X$. By Theorem \ref{two}, $h$ is uniformly continuous. It follows
(see earlier) that $h(D)$ is totally bounded; since $h(D)$ is dense in $X$, we
see that $X$ itself is totally bounded.
\end{proof}

\begin{cor}
\label{AA1a}Let $f$ be a strongly continuous mapping of a precompact space $X
$ into $\mathbf{R}$. Then $f(X)$ is totally bounded, and $\sup f,\inf f$ exist.
\end{cor}

\begin{proof}
By Proposition \ref{AA1}, $f(X)$ is a totally bounded subset of
$\mathbf{R}$. The rest of the proposition now follows from \cite{bvtech}
(Proposition 2.2.5).
\end{proof}

\begin{prop}
\label{AA2} A metric space is precompact if and only if it is totally bounded.
\end{prop}

\begin{proof}
Let $X$ be a totally bounded metric space, and $\widehat{X}$ its completion.
Since $X$ is dense in $\widehat{X}$, the latter is totally bounded; it is
therefore compact. By Theorem \ref{one}, there exists a uniformly
continuous---and hence strongly continuous---open mapping $h$ of
$2^{\mathbf{N}}$ onto $\widehat{X}$. Since $X$ is dense in $\widehat{X}$, the
openness of $h$ ensures that $h^{-1}(X)$ is dense in $2^{\mathbf{N}}$. Hence
$X$ is precompact.

The converse is a special case of Proposition \ref{AA1}.
\end{proof}%


We now digress to consider Cauchy nets and the extension of functions defined
on uniform subspaces\footnote{When we refer to a \textbf{uniform }(resp., \textbf{metric}%
)\textbf{\ subspace} $X$ of a pre-apartness space $Y$, we mean that there is a
uniform structure $\mathcal{U}$ (resp., metric $\rho$) on $X$ such that the
pre-apartness induced on $X$ by that on $Y$ coincides with the pre-apartness
induced by $\mathcal{U}$ (resp., $\rho$).} of apartness spaces.

Let $X$ be a pre-apartness space. We shall need the the \textbf{nested
neighbourhoods property,}

\begin{enumerate}[\textbf{(NN)}]
\item If $x\in-U$, then there exists $V\subset X$ such that
$x\in-V$ and $\lnot V\bowtie U$,
\end{enumerate}%

\noindent
and the property of \textbf{weak symmetric separatedness},

\begin{enumerate}[\textbf{(WSS)}]
\item If $S\bowtie T$, then for each $x\in X$ there exists
$U\subset X$ such that $x\in-U$ and $\lnot\left(  S-U\neq\varnothing\wedge
T-U\neq\varnothing\right)  $.
\end{enumerate}%

\noindent
Uniform spaces possess both of these properties.

A net $s\equiv\left(  x_{n}\right)  _{n\in\mathfrak{D}}$ in $X$ is said to be
\textbf{totally Cauchy} if for all subsets $A,B$ of $\mathfrak{D}$ with
$s(A)\bowtie s(B)$, there exists $N\in\mathfrak{D}$ such that%
\[
\lnot\left(  \exists_{m\in A}\left(  m\succcurlyeq N\right)  \wedge
\lnot\exists_{n\in B}\left(  n\succcurlyeq N\right)  \right)  ,
\]
where, as is the case throughout this paper, $\succcurlyeq$ is the preorder
relation on the index set $\mathfrak{D}$ of the net. Proposition 3.5.1 of
\cite{bvbook} says that the property \textbf{WSS} is equivalent to every
convergent net in $X$ being totally Cauchy. Proposition 3.5.15 of the same
book tells us that if $X$ has the property \textbf{NN}, and $s$ is a totally
Cauchy net in $X$ that has a subnet converging to a limit $x\in X$, then $s$
itself converges to $x$; and Theorem 3.5.12 that if $\left(  X,\mathcal{U}%
\right)  $ is a uniform (apartness) space, then every totally Cauchy net
$\left(  x_{n}\right)  _{n\in\mathfrak{D}}$ in $X$ with a countable subnet is
\textbf{uniformly Cauchy}, in the sense that%
\[
\forall_{U\in\mathcal{U}}\exists_{N\in\mathfrak{D}}\forall_{m,n\succcurlyeq
N}\left(  \left(  x_{m},x_{n}\right)  \in U\right)  .
\]

\begin{lem}
\label{F4}Let $\left(  Z,\bowtie\right)  $ be an apartness space, let $z\in Z$
have a countable base $\left(  V_{n}\right)  _{n\geq1}$ of neighbourhoods in
the apartness topology, and let $\left(  z_{n}\right)  _{n\in\mathfrak{D}}$ be
a net in $Z$ that converges to $z.$ Then there is a countable subnet $\left(
z_{n_{k}}\right)  _{k\geq1}$ of $\left(  z_{n}\right)  _{n\in\mathfrak{D}}$
that converges to $z.$
\end{lem}

\begin{proof}
There exists $n_{1}\in\mathfrak{D}$ such that $\left(  z,z_{n}\right)  \in
V_{1}$ for all $n\succcurlyeq n_{1}.$ Having constructed $n_{k}\in
\mathfrak{D},$ pick $\nu\succcurlyeq n_{k}$ such that $\left(  z,z_{n}\right)
\in V_{k+1}$ for all $n\succcurlyeq\nu.$ There exists $n_{k+1}\in\mathfrak{D}$
such that $n_{k+1}\succcurlyeq\nu$ and $n_{k+1}\succcurlyeq n_{k}.$ Then
$\left(  z,z_{n}\right)  \in V_{k+1}$ for all $n\succcurlyeq n_{k+1}.$ This
completes the inductive construction of a countable subnet $\left(  z_{n_{k}%
}\right)  _{k\geq1}$ such that $\left(  z,z_{n_{j}}\right)  \in V_{k}$
whenever $j\succcurlyeq k$. Since $\left(  V_{n}\right)  _{n\geq1}$ is a base
of neighbourhoods of $z$, it follows that $\left(  z_{n_{k}}\right)  _{k\geq1}
$ converges to $z$.
\end{proof}%


An apartness space is said to be \textbf{first countable} if each of its
elements has a countable base of neighbourhoods in the apartness topology.
First countability is used in several of the succeeding results, in order to
enable us to work with sequences rather than general nets.

\begin{lem}
\label{F5}Let $X$ be a dense uniform subspace of a first-countable, weakly
symmetrically separated apartness space $Y$, and let $\widehat{X}$ be a
complete uniform space in which $X$ is dense. Then there exists a mapping
$g:Y\rightarrow\widehat{X}$ such that

\begin{enumerate}[\em(i)]
\item $g(x)=x$ for each $x\in X$, and

\item for each net $\left(  x_{n}\right)  _{n\in\mathfrak{D}}$ of
elements of $X$ converging to $y\in Y$, the image net $\left(  g(x_{n}%
)\right)  _{n\in\mathfrak{D}}$ converges to $g(y)$.
\end{enumerate}
\end{lem}

\begin{proof}
Let $y\in Y,$ and construct a net $\left(  x_{n}\right)  _{n\in\mathfrak{D}}$
in $X$ that converges to $y$ in the apartness topology. Since $Y$ has the
property \textbf{WSS}, $\left(  x_{n}\right)  _{n\in\mathfrak{D}}$ is a
totally Cauchy net in $X$. By Lemma \ref{F4}, there is a countable subnet
$\left(  x_{n_{k}}\right)  _{k\geq1}$ of $\left(  x_{n}\right)  _{n\in
\mathfrak{D}}$ that converges to $y$. It follows from a remark before Lemma
\ref{F4} that the net $\left(  x_{n}\right)  _{n\in\mathfrak{D}}$ is uniformly
Cauchy in $X$. Since $\widehat{X}$ is a complete uniform space, this net
therefore converges to a unique limit $x_{\infty}\in\widehat{X}$. If $\left(
x_{m}^{\prime}\right)  _{m\in\mathfrak{M}}$ is another net in $X$ that
converges to $y$ in the apartness topology, then it, too, has a countable
subnet $\bigl(  x_{m_{j}}^{\prime}\bigr){}_{j\geq1}$ and converges to a
unique limit $x_{\infty}^{\prime}$ in $\widehat{X}$. Now, the sequences
$\left(  x_{n_{k}}\right)  _{k\geq1}$ and $\bigl(  x_{m_{j}}^{\prime}\bigr){}
_{j\geq1}$ both converge to $y$ in the space $Y$. Hence the sequence
$s\equiv\left(  x_{n_{1}},x_{m_{1}}^{\prime}x_{n_{2}},x_{m_{2}}^{\prime
},\ldots\right)  $ converges to $y$. Since $y$ has the property \textbf{WSS},
this sequence is totally Cauchy in $X$; but $X$ is a uniform space, so, being
itself a countable net, the sequence is uniformly Cauchy (again see the last
remark before Lemma \ref{F4}). It therefore converges to a unique limit in
$\widehat{X}$. But it has a subsequence converging to $x_{\infty}$ and a
subsequence converging to $x_{\infty}^{\prime}$. Since $\widehat{X}$, being a
uniform space, has the property \textbf{NN}, another remark before Lemma
\ref{F4} shows that $s$ converges to both $x_{\infty}$ and $x_{\infty}%
^{\prime}$. By the uniqueness of limits in a uniform space, $x_{\infty
}^{\prime}=x_{\infty}$. It now follows that setting $g(y)=x_{\infty}$ gives us
a good definition of a function $g:Y\rightarrow\widehat{X}$ with properties
(i) and (ii).
\end{proof}%


An extremely important property applicable to a pre-apartness space $X$ is the
\textbf{Efremovi\v{c} property} (a form of topological normality):

\begin{enumerate}[\textbf{(EF)}]
\item $S\bowtie T\Rightarrow\exists_{E\subset X}\left(
A\bowtie\lnot E\wedge E\bowtie B\right)  $.
\end{enumerate}%

\noindent
This property implies that%
\begin{equation}
\forall_{A,B\subset X}\left(  A\bowtie B\Leftrightarrow\lnot\lnot A\bowtie
B\right)  ,\label{D11}%
\end{equation}
and that%
\begin{equation}
\forall_{A,B\subset X}\left(  A\bowtie B\Leftrightarrow\overline{A}%
\bowtie\overline{B}\right) \label{D12}%
\end{equation}
(\cite{bvbook}, Propositions 3.1.10, 3.1.7, 3.1.9, and Corollary 3.1.19).

\begin{lem}
\label{F6}Let $\left(  Y,\bowtie\right)  $ be a first-countable, weakly
symmetrically separated apartness space with the Efremovi\v{c} property, let
$X$ be a uniform subspace of $Y$, and let $\widehat{X}$ be a complete uniform
space in which $X$ is dense. Then there exists a strongly continuous mapping
$g$ of $Y$ onto a dense subset of $\widehat{X},$ such that $g(x)=x$ for each
$x\in X$.
\end{lem}

\begin{proof}
Construct the mapping $g$ as in Lemma \ref{F5}. Given $y\in Y$, define%
\[
\mathfrak{D}_{y}\equiv\left\{  \left(  x,U\right)  :x\in X,y\in-U,x\in
-U\right\}
\]
and define the preorder $\succcurlyeq$ on $\mathfrak{D}_{y}$ by%
\[
\left(  x,U\right)  \succcurlyeq\left(  x^{\prime},U^{\prime}\right)
\Leftrightarrow-U\subset-U^{\prime}.
\]
Then $\left(  x,U\right)  $ $\rightsquigarrow x$ is a net in $X$ that
converges to $y$ in the space $Y$. We denote this net by $\left(
x_{n}\right)  _{n\in\mathfrak{D}_{y}}$; so if $n=\left(  x,U\right)
\in\mathfrak{D}_{y}$, then $x_{n}=x$. Note that, by the definition of the
mapping $g$, this net converges to $g(y)$ in the space $\widehat{X}$.

The proof that $g$ is strongly continuous uses an argument very similar to one
in the proof of Theorem 3.5.22 of \cite{bvbook}; we omit the details. Finally,
since $g(Y)$ contains $X$, it is dense in $\widehat{X}$.
\end{proof}%


We are now able to give a limited generalisation of the theorem that a dense
subspace of a totally bounded uniform space is totally bounded.

\begin{prop}
\label{F7}Let $Y$ be a first-countable, weakly symmetrically separated,
precompact apartness space with the Efremovi\v{c} property, let $X$ be a dense
uniform subspace of $Y$, and let $\widehat{X}$ be a complete uniform space in
which $X$ is dense. Then $X$ is totally bounded.
\end{prop}

\begin{proof}
By Lemma \ref{F6}, there exists a strongly continuous mapping $g$ of $Y$ onto
a dense subspace $Z$ of $\widehat{X}$ such that $g(x)=x$ for each $x\in X$.
Then $g(Y)$, and therefore $Z$, is precompact; whence, by Proposition
\ref{AA1}, $Z$ is totally bounded. It follows that $X$, being clearly dense in
$Z$, is totally bounded.
\end{proof}

\begin{cor}
\label{AA3}Let $Y$ be a first-countable, weakly symmetrically separated,
precompact apartness space with the Efremovi\v{c} property, and $X$ a dense
metric subspace of $Y$. Then $X$ is totally bounded and precompact.
\end{cor}

\begin{proof}
In Proposition \ref{F7}, take $\widehat{X}$ to be the metric completion of $X
$, to show that $X$ is totally bounded. Reference to Proposition \ref{AA2}
shows that $X$ is precompact.
\end{proof}

\section{Unions and products of precompact spaces}

Bearing in mind what happens with total boundedness in metric spaces, we would
hope that precompactness passes to finite unions and products of apartness
subspaces. To deal with unions, we need a construction that glues two
apartness spaces together to make a new apartness space.

Let $X,Y$ be (symmetric) pre-apartness spaces, and let%
\[
X\Cup Y\equiv\left(  X\times\left\{  0\right\}  \right)  \cup\left(
Y\times\left\{  1\right\}  \right)  ,
\]
with equality and inequality those induced by their standard counterparts on
$X\times Y$. Subsets of $X\Cup Y$ have the form%
\[
A\equiv\left(  A^{0}\times\left\{  0\right\}  \right)  \cup\left(  A^{1}%
\times\left\{  1\right\}  \right)
\]
with $A^{0}\subset X$ and $A^{1}\subset Y$. We prove that the relation
$\bowtie$ defined for subsets $A,B$ of $X\Cup Y$ by%
\[
A\bowtie_{X\Cup Y}B\Leftrightarrow\left(  A^{0}\bowtie_{X}B^{0}\text{ }%
\wedge\text{ }A^{1}\bowtie_{Y}B^{1}\right)
\]
is an apartness, where such notations as $\bowtie_{X}$ have the obvious interpretation.

Clearly, $\bowtie_{X\Cup Y}$ is symmetric and satisfies \textbf{(B1)}. Before
verifying \textbf{(B2)}, we note that for any $A\subset X\Cup Y$ and for
$i\in\{1,2\}$,%
\[
\left(  -A\right)  ^{i}=-A^{i}\text{ and }\left(
\mathord{\sim}%
A\right)  ^{i}=%
\mathord{\sim}%
A^{i}.
\]
For example,\label{comp}%
\begin{align*}
\left(  x,0\right)  \in-A  & \Leftrightarrow\left(  x,0\right)  \bowtie_{X\Cup
Y}A\\
& \Leftrightarrow\left\{  \left(  x,0\right)  \right\}  ^{0}\bowtie
A^{0}\text{ and }\left\{  \left(  x,0\right)  \right\}  ^{1}\bowtie A^{1}\\
& \Leftrightarrow x\bowtie A^{0}\text{ and }\varnothing\bowtie A^{1}\\
& \Leftrightarrow x\in-A^{0}\text{.}%
\end{align*}
Thus if $\left\{  \left(  x,0\right)  \right\}  \bowtie_{X\Cup Y}A$ in $X\Cup
Y$, then, by \textbf{(B2)} in $X$, $x\in%
\mathord{\sim}%
A^{0}$ and therefore $\left(  x,0\right)  \in%
\mathord{\sim}%
\left(  A^{0}\times\left\{  0\right\}  \right)  $. Since, clearly, $\left(
x,0\right)  \in$ $%
\mathord{\sim}%
\left(  A^{1}\times\left\{  1\right\}  \right)  $, we conclude that $\left(
x,0\right)  \in%
\mathord{\sim}%
A$. Likewise, if $\left(  y,1\right)  \bowtie_{X\Cup Y}A$ in $X\Cup Y$, then
$\left(  y,1\right)  \in%
\mathord{\sim}%
A$. This completes the proof that \textbf{(B2)} is satisfied.

For \textbf{(B3)}, let $A_{1},A_{2}\subset X$ and $B_{1},B_{2}\subset Y$. Then%
\begin{align*}
\left(  A_{1}\cup A_{2}\right)  \bowtie_{X\Cup Y}\left(  B_{1}\cup
B_{2}\right)   & \Leftrightarrow\forall_{i\in\left\{  0,1\right\}  }\left(
\left(  A_{1}\cup A_{2}\right)  ^{i}\bowtie\left(  B_{1}\cup B_{2}\right)
^{i}\right) \\
& \Leftrightarrow\forall_{i\in\left\{  0,1\right\}  }\left(  A_{1}^{i}\cup
A_{2}^{i}\bowtie B_{1}^{i}\cup B_{2}^{i}\right) \\
& \Leftrightarrow\forall_{i\in\left\{  0,1\right\}  }\forall_{j,k\in\left\{
1,2\right\}  }\left(  A_{j}^{i}\bowtie B_{k}^{i}\right) \\
& \Leftrightarrow\forall_{j,k\in\left\{  1,2\right\}  }\forall_{i\in\left\{
0,1\right\}  }\left(  A_{j}^{i}\bowtie B_{k}^{i}\right) \\
& \Leftrightarrow\forall_{j,k\in\left\{  1,2\right\}  }\left(  A_{j}%
\bowtie_{X\Cup Y}B_{k}\right)  .
\end{align*}
To handle \textbf{(B4)}, let $-A\subset%
\mathord{\sim}%
B$ in $X\Cup Y$. Then $-A^{i}=\left(  -A\right)  ^{i}\subset\left(
\mathord{\sim}%
B\right)  ^{i}=%
\mathord{\sim}%
B^{i}$. Applying \textbf{(B4)} in $X$ and $Y$, we obtain $-A^{i}\subset-B^{i}$
and therefore $\left(  -A\right)  ^{i}\subset\left(  -B\right)  ^{i}$. Hence
$-A\subset-B$ in $X\Cup Y$.

Finally, consider $\left(  x,0\right)  \in-A$ in $X\Cup Y$. By the foregoing,
$x\in-A^{0}$; whence, by \textbf{(B5)} in $X$, there exists $S^{0}\subset X$
such that $x\in-S^{0}$ and $X=-A^{0}\cup S^{0}$. Taking $S^{1}=Y$, we see that%
\[
\left\{  \left(  x,0\right)  \right\}  ^{0}=\left\{  x\right\}  \bowtie_{X\Cup
Y}S^{0}\times\left\{  0\right\}
\]
and%
\[
\left\{  \left(  x,0\right)  \right\}  ^{1}=\varnothing\bowtie_{X\Cup Y}%
S^{1}\times\{1\}.
\]
Hence $x\in-S$, where%
\[
S\equiv\left(  S^{0}\times\left\{  0\right\}  \right)  \cup\left(  S^{1}%
\times\{1\}\right)  .
\]
Moreover, if $\left(  x^{\prime},0\right)  \in X\Cup Y$, then $x^{\prime}%
\in-A^{0}\cup S^{0}=\left(  -A\right)  ^{0}\cup S^{0}$; so either $\left(
x^{\prime},0\right)  \in-A$ or $\left(  x^{\prime},0\right)  \in S^{0}%
\times\left\{  0\right\}  $. On the other hand, if $\left(  y,1\right)  \in
X\Cup Y$, then%
\[
\left(  y,1\right)  \in Y\times\left\{  1\right\}  =S^{1}\times\{1\}\subset
S.
\]
Thus $X\Cup Y=-A\cup\left(  S^{0}\times\left\{  0\right\}  \right)  $. A
similar argument, starting with $\left(  y,1\right)  \in-A$ in $X\Cup Y$,
completes the verification of \textbf{(B5)} and hence the proof that
$\bowtie_{X\Cup Y}$ is an apartness on $X\Cup Y$.

We call the space $X\Cup Y$, taken with the apartness we have just
constructed, the \textbf{disjoint union of the apartness spaces} $X$ and $Y$
(in that order).

\begin{lem}
\label{0103aa1}Let $X,Y$ be subspaces of an apartness space $E$. Then the
mapping $g$ of $X\Cup Y$ onto $X\cup Y$ defined by%
\begin{align*}
g(x,0)  & \equiv x\ \ \ \ \left(  x\in X\right)  ,\\
g(y,1)  & \equiv y\ \ \ \ \left(  y\in Y\right)
\end{align*}
is strongly continuous.
\end{lem}

\begin{proof}
Let $A\bowtie B$ in the subspace $X\cup Y$ of $E$. Then%
\[
\left(  g^{-1}(A)\right)  ^{0}=\left(  A\cap X\right)  \bowtie_{E}\left(
B\cap X\right)  =\left(  g^{-1}(B)\right)  ^{0}
\]
and similarly $\left(  g^{-1}(A)\right)  ^{1}\bowtie_{E}\left(  g^{-1}%
(B)\right)  ^{1}$. Hence $g^{-1}(A)\bowtie_{X\Cup Y}g^{-1}(B)$.
\end{proof}%


The foregoing construction and lemma enable us to prove

\begin{prop}
\label{0103aa2}The union of two precompact subsets of an apartness space is precompact.
\end{prop}

\begin{proof}
Let $X_{0},X_{1}$ be precompact subsets of an apartness space $E$. Then for
each $k,$ there exist a dense subset $D_{k}$ of $2^{\mathbf{N}}$ and a
strongly continuous mapping $h_{k}$ of $D_{k}$ onto a dense subset of $X_{k}$.
Define a mapping%
\[
w:\left(  D_{1}\times\left\{  0\right\}  \right)  \cup\left(  D_{2}%
\times\{1\}\right)  \rightarrow X_{0}\Cup X_{1}
\]
by%
\[
w(\alpha,0)\equiv\left(  h_{1}(\alpha),0\right)  \text{ and }w(\alpha
,1)\equiv\left(  h_{2}(\alpha),1\right)  .
\]
We first prove that%
\[
S\equiv\left(  h_{1}(D_{1})\times\left\{  0\right\}  \right)  \cup\left(
h_{2}(D_{2})\times\{1\}\right)  ,
\]
the range of $w$, is dense in the apartness space $X_{0}\Cup X_{1}$. If $x\in
X_{0}$ and $\left(  x,0\right)  \in-A$ in $X_{0}\Cup X_{1}$, then (as we
showed on page \pageref{comp}) $x\in-A^{0}$ in $X_{0}$. Since $h_{1}(D_{1})$
is dense in $X_{0}$, there exists $\alpha\in D_{1}$ such that $h_{1}%
(\alpha)\in-A^{0}$; since $\left\{  \left(  h_{1}(\alpha),0\right)  \right\}
^{1}=\varnothing\bowtie A^{1}$ in $X_{1}$, we now see that $\left\{  \left(
h_{1}(\alpha),0\right)  \right\}  \bowtie_{X_{0}\Cup X_{1}}A$ and hence that
$\left(  h_{1}(\alpha),0\right)  \in-A$. Thus $h_{1}(D_{1})\times\left\{
0\right\}  $ is dense in $X_{0}\times\left\{  0\right\}  $ relative to the
apartness topology on $X_{0}\Cup X_{1}$. Likewise, $h_{2}(D_{2})\times\{1\}$
is dense in $X_{1}\times\{1\}$. Hence $S$ is dense in $X_{0}\Cup X_{1}$.

To prove that $w$ is strongly continuous, let $A\bowtie_{X_{0}\Cup X_{1}}B$.
Then $A^{0}$ $\bowtie B^{0}$ in $X_{0}$, so $h_{1}^{-1}(A^{0})\bowtie
h_{1}^{-1}(B^{0})$ in $D_{1}$. The definition of the product apartness on
$2^{\mathbf{N}}\times\left\{  0,1\right\}  $ (see Section 3.7 of
\cite{bvbook}) now yields%
\[
w^{-1}\left(  A^{0}\times\left\{  0\right\}  \right)  =h_{1}^{-1}(A^{0}%
)\times\left\{  0\right\}  \bowtie h_{1}^{-1}(B^{0})\times\left\{  0\right\}
=w^{-1}(B^{0}\times\left\{  0\right\}  ).
\]
On the other hand, denoting the $k$th projection of the Cartesian product of
two sets by $\mathrm{pr}_{k}$, we have
\begin{align*}
w^{-1}\left(  A^{0}\times\left\{  0\right\}  \right)   & \subset2^{\mathbf{N}%
}\times\left\{  0\right\}  ,\\
w^{-1}\left(  B^{1}\times\left\{  1\right\}  \right)   & \subset2^{\mathbf{N}%
}\times\left\{  1\right\}  ,\text{ and}\\
\mathrm{pr}_{2}\left(  2^{\mathbf{N}}\times\left\{  0\right\}  \right)   &
\bowtie\mathrm{pr}_{2}\left(  2^{\mathbf{N}}\times\left\{  1\right\}  \right)
;
\end{align*}
whence, by definition of the product apartness on $2^{\mathbf{N}}%
\times\left\{  0,1\right\}  $ , $w^{-1}(A^{0}\times\left\{  0\right\}
)\bowtie w^{-1}(B^{1}\times\left\{  1\right\}  )$. Thus, by \textbf{(B3)} in the
metric space $2^{\mathbf{N}}\times\left\{  0,1\right\}  $,
\begin{align*}
w^{-1}(A^{0}\times\left\{  0\right\}  )  & \bowtie w^{-1}\left(  B^{0}%
\times\left\{  0\right\}  \right)  \cup w^{-1}(B^{1}\times\left\{  1\right\}
)\\
& =w^{-1}\left(  \left(  B^{0}\times\left\{  0\right\}  \right)  \cup
(B^{1}\times\left\{  1\right\}  )\right)  =w^{-1}(B).
\end{align*}
Similarly, $w^{-1}(A^{1}\times\left\{  1\right\}  )\bowtie w^{-1}(B)$; whence,
by another application of \textbf{(B3)} in $2^{\mathbf{N}}\times\left\{
0,1\right\}  $, we obtain $w^{-1}(A)\bowtie w^{-1}(B)$. Thus $w$ is strongly
continuous. Now, $2^{\mathbf{N}}\times\left\{  0,1\right\}  $, being the
product of two compact metric spaces, is a compact metric space; so, by
Theorem \ref{one}, there exists a uniformly---and hence strongly---continuous
open mapping $f$ of $2^{\mathbf{N}}$ onto $2^{\mathbf{N}}\times\left\{
0,1\right\}  $. Since $\left(  D_{1}\times\left\{  0\right\}  \right)
\cup\left(  D_{2}\times\{1\}\right)  $ is dense in $2^{\mathbf{N}}%
\times\left\{  0,1\right\}  $ and $f$ is an open mapping,%
\[
D\equiv f^{-1}\left(  \left(  D_{1}\times\left\{  0\right\}  \right)
\cup\left(  D_{2}\times\{1\}\right)  \right)
\]
is dense in $2^{\mathbf{N}}$. Moreover, the restriction of $w\circ f$ to $D$
is a strongly continuous mapping of $D$ onto the dense subset $S$ of
$X_{0}\Cup X_{1}$. On the other hand, the mapping $g$ of $X_{0}\Cup X_{1}$
onto $X_{0}\cup X_{1}$ defined in Lemma \ref{0103aa1} is strongly continuous.
Hence $g\circ w\circ f$ is a strongly continuous mapping of the dense subset
$D$ of $2^{\mathbf{N}}$ onto the dense subset $h_{1}(D_{1})\cup h_{2}(D_{2})$
of $X_{0}\cup X_{1}$. Thus $X_{0}\cup X_{1}$ is precompact.
\end{proof}%


We now pass from unions to products.

\begin{prop}
\label{0103aa4}The product of two apartness spaces is precompact if and only
if each of the factors is precompact.
\end{prop}

\begin{proof}
Let $X_{0},X_{1}$ be apartness spaces, and $X$ their product apartness space.
If $X$ is precompact, then, since the projection mappings from $X$ onto its
factors are strongly continuous (\cite{bvbook}, Corollary 3.7.3) and hence preserve precompactness, we see that $X_{0},X_{1}$ are precompact.

Suppose, conversely, that the two factor spaces are precompact. For each $k$,
there exist a dense subset $D_{k}$ of $2^{\mathbf{N}}$ and a strongly
continuous mapping $h_{k}$ of $D_{k}$ onto a dense subset of $X_{k}$. Then, by
Proposition 3.7.6 of \cite{bvbook},%
\[
(h_{1},h_{2}):\left(  \alpha_{1},\alpha_{2}\right)  \rightsquigarrow\left(
h_{1}(\alpha_{1}),h_{2}(\alpha_{2})\right)
\]
is a strongly continuous mapping of the dense subspace $D_{1}\times D_{2}$ of
$2^{\mathbf{N}}\times2^{\mathbf{N}}$ onto the dense subspace $h_{1}%
(D_{1})\times h_{2}(D_{2})$ of $X_{0}\times X_{1}$, which is therefore
precompact.%
\end{proof}%


This completes our introductory look at a possible generalisation of total
boundedness in the context of an arbitrary apartness space. While most of the
preceding results are exactly what one might wish for, there remain at least
two undesirable features of the theory:

\begin{enumerate}[$-$]
\item the lack of a general analogue of the proposition that if a uniform
space is totally bounded, then so are its dense subsets (Proposition \ref{F7}
being a very restricted version of such an analogue);

\item the inability to replace \emph{metric space} by \emph{separable
uniform space} in Proposition \ref{AA2}.
\end{enumerate}%

\noindent
Moreover, since $2^{\mathbf{N}}$ is separable, so is each precompact space.
Nevertheless, precompactness may well be worthy of further constructive exploration.%

%

%

\noindent
\textbf{Acknowledgement.} \ This work was done in part when the author was a
guest of Professor Helmut Schwichtenberg in the Logic Section of the
Mathematisches Institut, Ludwig-Maximilians-Universit\"{a}t, M\"{u}nchen.%


%

\end{document}